\documentclass[11pt,reqno]{amsart}
\usepackage{amssymb}

\usepackage{amsthm,amsfonts,amssymb,euscript,hyperref,color}
\usepackage{comment}
\usepackage{mathtools}
\usepackage{stmaryrd}
\usepackage{mathrsfs}
\usepackage{pifont}
\usepackage{amsfonts,amsmath,amssymb,amsthm,amscd,cancel}
\usepackage{graphicx}
\usepackage[]{geometry}
\usepackage{verbatim}
\usepackage{mathrsfs}
\usepackage{fancyhdr}
\usepackage{footnpag,footmisc}

\usepackage[displaymath]{lineno}
\theoremstyle{definition}
\newtheorem{theorem}{Theorem}[section]

\newtheorem{remark}{Remark}[section]

\allowdisplaybreaks[4]

\DeclareMathAlphabet{\mathsfsl}{OT1}{cmss}{m}{sl}
\numberwithin{equation}{section}

\newcommand{\D}{\mathrm{d}}

\def\Lb{\underline{L}}

\def\omegab{\underline{\omega}}

\def\ub{\underline{u}}
\def\Cb{\underline{C}}

\newcommand{\Db}{\underline{D}}

\newcommand{\hb}{\underline{h}}

\begin{document}

\title[Interior instability]{Interior instability of naked singularities of a scalar field}

\author[Junbin Li]{Junbin Li}
\address{Department of Mathematics, Sun Yat-sen University\\ Guangzhou, China}
\email{lijunbin@mail.sysu.edu.cn}

\date{}

\maketitle

\begin{abstract}
We show that the $k$-self-similar naked singularity solutions of the spherically symmetric Einstein--Scalar field system are unstable to black hole formation under perturbations that are totally supported in the interior region, in all regularities strictly below the threshold.  The instability below the threshold is also established for exterior perturbations. We also show that general naked singularity solutions are unstable under interior BV perturbations, which provides a new  insight  into understanding the weak cosmic censorship conjecture for this model.  In contrast to all previous results on the exterior instability of naked singularities (and even trapped surface formation), where only a single incoming null cone is considered, the novel approach to proving the interior instability is analyzing a family of incoming null cones becoming more and more singular.
\end{abstract}


\setcounter{tocdepth}{1}


\allowdisplaybreaks

\section{Introduction}

\subsection{Preliminaries} One of the most important problems in general relativity is the weak cosmic censorship conjecture, which states that naked singularities do not \emph{generically}  form in gravitational collapse. Here gravitational collapse refers to the solution (or precisely, Cauchy development) of the Einstein equations with asymptotically flat initial data. Loosely speaking, naked singularities are singularities that are not hidden inside a black hole, whose future null cone extends to asymptotic region in the spacetime, so that the world-lines of far away observers will intersect this null cone, which means that the observers receive signals from naked singularities. In a precise mathematical language, the weak cosmic censorship states that generic asymptotically flat initial data lead to a maximal Cauchy development possessing a complete future null infinity (see \cite{Chr99cqg} for the precise definition of the completeness of the future null infinity without using the language of conformal compactification), while a solution with a naked singularity means an asymptotically flat Cauchy development without complete future null infinity (also see \cite{R-Sh} for a precise definition of this).
 
   The original version of the weak cosmic censorship conjecture proposed by Penrose in \cite{Pen69} states that NO naked singularities can form and the end state of gravitational collapse should be black hole formation.  However, as examples of naked singularity solutions in Einstein--scalar field system were constructed (\cite{Cho, Chr94}), and more naked singularity solutions were constructed for other matter fields and in vacuum (\cite{R-Sh, O-P, G-H-J}), the word ``generically'' in the current version of the weak cosmic censorship conjecture is necessary. One of the strategies for proving the weak cosmic censorship conjecture would then be showing that all possible naked singularity solutions are unstable under perturbations from initial data. 

In this paper we consider the spherically symmetric solutions of Einstein--scalar field system,   a couple system of a Lorentzian manifold $(\mathcal{M},g)$ and a real scalar function $\phi$, 
\begin{align}\label{ES}
\mathbf{Ric}_{\alpha\beta}-\frac{1}{2}\mathbf{R}g_{\alpha\beta}=2\mathbf{T}_{\alpha\beta}\end{align}
where
\begin{align*}
\mathbf{T}_{\alpha\beta}=\nabla_\alpha\phi\nabla_\beta\phi-\frac{1}{2}g_{\alpha\beta}g^{\mu\nu}\nabla_\mu\phi\nabla_\nu\phi,
\end{align*}
and $\phi$ obeys linear wave equation
\begin{align}\label{Wave}
\Box_g\phi=0.
\end{align}
As the Birkhoff theorem eliminates gravitational dynamics in spherical symmetry, the (spherically symmetric) Einstein--scalar field is the simplest toy model to gain insights into the vacuum Einstein equations. Christodoulou studied this model extensively (including papers \cite{Chr87, Chr91, Chr93,Chr94,Chr99}). He constructed self-similar naked singularity solutions (\cite{Chr91}), which is the main object considered in this paper, obtained a sharp extension principle for $C^1$ solution, studied extensively the concept of BV solution (\cite{Chr93}) and based on this, he showed that (\cite{Chr91, Chr99}) all naked singularities arising in the system \eqref{ES} are unstable to black hole formation under BV perturbations, which is the scale invariant $L^1$ based topology in spherical symmetry.

\begin{figure}
\includegraphics[width=2.4 in]{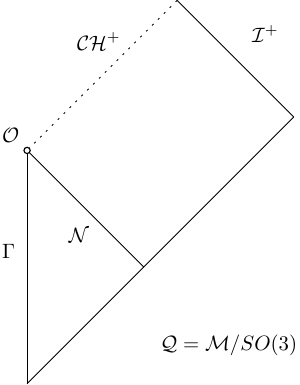}
\caption{Penrose diagram of a spherically symmetric naked singularity solution}
\label{fig:nakedsingularity}
\end{figure}

The Penrose diagram of a spherically symmetric naked singularity solution can be depicted as in Figure \ref{fig:nakedsingularity}. It is the Cauchy development of some initial data set given on an outgoing null cone from a point. It terminates at a first singularity $\mathcal{O}$ at the center $\Gamma$. The future null cone of $\mathcal{O}$ is called Cauchy horizon $\mathcal{CH}^+$, which intersects the future null infinity $\mathcal{I}^+$. The spacetime has no trapped region. The interior region of a naked singularity solution refers to the region to the past of the singularity, and the boundary of the interior region is the past null cone $\mathcal{N}$ of $\mathcal{O}$. The exterior region  refers to the future of $\mathcal{N}$. The instability result proved in \cite{Chr99} (and also an alternative   proof without contradiction argument in \cite{Liu-Li}) is the following.  Given that the mass radio $\mu=\frac{2m}{r}\nrightarrow0$ as approaching $\mathcal{O}$ along its past null cone $\mathcal{N}$ (which holds true for all self-similar examples constructed in \cite{Chr94} where $\mu$ is a non-zero constant along $\mathcal{N}$) , and as a consequence, the blue shift 
\begin{equation}\label{blueshift}\int_{\mathcal{N}}\frac{\mu}{1-\mu}\frac{1}{r}\D r\end{equation}
being infinite,  a precise small BV perturbation\footnote{By BV perturbations  we mean the perturbation is small in BV topology, that is, the total variation of the difference is small. But the perturbations itself can have higher regularity, like AC, absolute continuous.} can be constructed, so that a black hole forms before the singularity (see Figure \ref{fig:nakedperturbation}). This family of perturbations is supported in the exterior region. In the constructed proof in \cite{Liu-Li}, we use instead
\begin{equation}\label{blueshift2}\int_{\mathcal{N}}r\left(\frac{n\phi}{nr}\right)^2\D r\end{equation}
 which is also infinite and related to the null lapse, where $n$ in the incoming null direction.

\begin{figure}
\includegraphics[width=2.8 in]{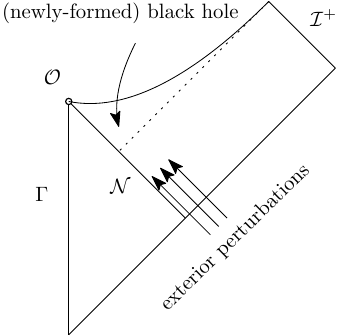}
\caption{Exterior instability to black hole formation}
\label{fig:nakedperturbation}
\end{figure}

\begin{remark}
The infiniteness of \eqref{blueshift2} is easier to utilize in non-spherically symmetric setting, see for example \cite{Li-Liu1, A}, in which non-spherically symmetric perturbations of spherically symmetric naked singularity of a scalar field are considered. In \cite{Liu-Li} we showed that the infiniteness of \eqref{blueshift} implies the infiniteness of \eqref{blueshift2}. See also Remark \ref{blueshiftequal} for the relation between these two integrals.
\end{remark}

The problem then arises what will happen if we add interior perturbations, by which we mean perturbations not supported in the exterior region.   One of the main reasons for studying interior perturbations is to find dispersive perturbations, which can be seen in numerical studies in critical collapse of scalar field (see \cite{Cho}, or a comprehensive survey \cite{G-H-M} which includes both numerical and mathematical results). Here by dispersive perturbations we mean perturbations under which the singularity disappears and the solution becomes global. Dispersive perturbations cannot be found in exterior perturbations since they will not affect the singularity by finite speed of propagation, and in fact one of the main difficulties in studying interior perturbations is that the singularity is affected by interior perturbations. It is worth mentioning that the critical collapse naked singularity, which is discretely self-similar, are significantly different from the examples constructed in \cite{Chr94}, which we may call ``continuously'' self-similar solutions. It is not known whether  dispersive perturbations exist for continuously self-similar naked singularity solutions of scalar field even in numerical level. On the other hand, it is also possible that interior perturbations lead to black hole formation (which is also suggested in numerical studies of critical collapse).

In this paper we begin the study of interior instability, by studying the H\"older interior instability of continuously self-similar, or precisely, $k$-self-similar naked singularities constructed in \cite{Chr94} and then generalize the idea to study interior BV instability of general naked singularity solutions.  These perturbations all lead to  black hole formation. 

\subsection{The main results for $k$-self-similar naked singularities} 

Let us first review some basic facts about the $k$-self-similar solutions, more explicit information is presented in Section \ref{section:kselfsimilar}.   A spherically symmetric spacetime $(\mathcal{M},g)$ means that $SO(3)$ acts isometrically on it and $\mathcal{M}/SO(3)$ will be a $2$-d Lorentzian manifold with boundary $\Gamma$, which is called the center of the symmetry. The area radius of the orbits, denoted by $r$, is a geometric function in spherically symmetric spacetimes. A spherically symmetric solution $(\mathcal{M}, g, \phi)$ to \eqref{ES} should also carry a function $\phi$ obeying \eqref{Wave}, which is also spherically symmetric in the sense that $\phi$ is constant on each orbit.

A $k$-self-similar solution ($k\ne0$) considered in \cite{Chr94} is the following. Let $(\mathcal{M},g)$ be a spherically symmetric solution to \eqref{ES} and $h$ be the quotient metric on $\mathcal{Q}=\mathcal{M}/SO(3)$. Let $r$ be the area radius of the orbits (considered as a function on $\mathcal{Q}$), then $(\mathcal{Q},h,r,\phi)$ is called $k$-self-similar if there exists a one-parameter group of diffeomorphism  $f_a$ for $a>0$ on $\mathcal{Q}$ such that
$$f_a^*h=a^2h, \ f_a^*r=ar,\  f_a^*\phi=\phi-k\log a.$$
The vector field $S$ generated by $f_a$ obeys
$$\mathcal{L}_Sh=2h,\ Sr=r,\ S\phi=-k$$
and is then conformal Killing
 $$\mathcal{L}_Sg=2g.$$
 To study such solutions,  Christodoulou introduced self-similar Bondi coordinates $(u,r), u<0, r>0$, in which the quotient metric reads
\begin{equation}\label{Bondimetric}h=-\mathrm{e}^{2\nu}\D u^2-2\mathrm{e}^{\nu+\lambda}\D u\D r,\end{equation}
where $\nu,\lambda$ and $\phi+k\log(-u)$ are functions of $-\frac{r}{u}$ only. After a suitable rescaling on $u$, the conformal vector field $S$ generated by $f_a$ has the form
$$S=u\frac{\partial}{\partial u}+r\frac{\partial}{\partial r}.$$
For $k^2\in(0,\frac{1}{3})$, Christodoulou was able to show the existence of $k$-self-similar naked singularity solutions, whose Penrose diagram is depicted in Figure \ref{fig:nakedsingularity}. The singularity $\mathcal{O}$ corresponds to $u=r=0$. Recently Le \cite{Le} showed that these $k$-self-similar solutions cannot extend beyond $\mathcal{O}$ in $C^{0,1}$.

 A very important feature of the $k$-self-similar solutions is the finite regularity along the past null cone of the singularity $\mathcal{O}$, which is significantly different from the discretely self-similar critical collapse solution. The critical collapse solution is expected to be smooth but its rigorous  construction is still open (see \cite{Ci-Ke} for some progress in exterior constructions). Precisely speaking, in these $k$-self-similar naked singularities with self-similar Bondi coordinates, $r\frac{\partial\phi}{\partial r}$, the outgoing null derivative parametrized by $r$, only has a finite H\"older regularity $C^{\frac{k^2}{1-k^2}}$ in  $r$ along the past null cone of $\mathcal{O}$\footnote{Note that this finite regularity statement is written in the geometric function $r$, and hence is coordinate independent.} and at least $C^1$ elsewhere. A pioneer work \cite{R-Sh} by Rodnianski--Shlapentokh-Rothman develops methods to establish the stability of naked singularities under high-regularity exterior perturbations. In particular they use these methods to construct exterior of naked singularity solutions in vacuum (also see Singh's papers \cite{JS1, JS2}). Based on these methods, one can show that (see detailed discussions in \cite{JS2}) all exterior small nonlinear perturbations in a specific form (see Figure \ref{fig:nakedstability}), which is in $C^{\frac{\alpha}{1-\alpha}}$, lead to stability, for any $\alpha\ge k^2$, and $\frac{k^2}{1-k^2}$ serves as a threshold regularity.  
 
\begin{figure}
\includegraphics[width=2.6 in]{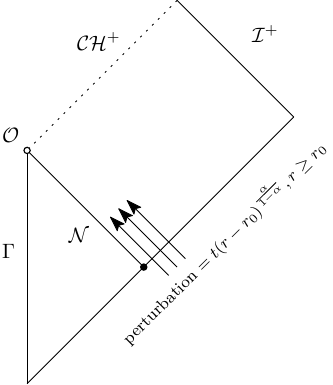}
\caption{Exterior stability under high regularity perturbations. The initial cone is parametrized by $r$ and $r_0$ is the value of $r$ at the its intersection of $\mathcal{N}$.}
\label{fig:nakedstability}
\end{figure}
 
 For interior perturbations, studies have begun by Singh in \cite{JS2}, in which the poor-man's linearization 
 $$\Box_{g_k}\phi=0$$
  was considered, where $g_k$ is a (perturbations of) $k$-self-similar solution for $k^2\ll1$. It was shown that the spherically symmetric solutions have stability (which means that the solution has self-similar bound as approaching the singularity, see \cite{JS2} for precise statement) if the initial data supported in the interior has regularity equal and above the threshold and instability below the threshold.  
However, the generalizations of the above analysis to nonlinear setting remain to be understood.

 The main content of the first theorem of this paper is to study nonlinear interior instability of the  $k$-self-similar naked singularity solutions. We identify a large class of interior perturbations leading to black hole formation. Moreover, we establish instability in all regularities below the threshold, consistent with the linear analysis in \cite{JS2}. Perhaps except for the threshold case, this instability is optimal. 

To make our statement more precise, given a $k$-self-similar naked singularity solution, we choose some outgoing null cone $C_{-1}^+$ of a point at the center $\Gamma$, parametrized by the area radius $r$, serving as the initial null cone. The initial data on $C_{-1}^+$ consists of $\alpha=\partial_r(r\phi)\big|_{C_{-1}^+}$ on the half real line.  In the $k$-self-similar naked singularity solutions, $\alpha$ is $C^{\frac{k^2}{1-k^2}}$ and $C^1$ except at one point.

\begin{remark}
 Christodoulou introduced the concept of BV solution in \cite{Chr93} as a suitable class of solution in which we can talk about the weak cosmic censorship. The term BV refers that $\alpha\in BV[0,+\infty)$.  We will not present the precise definition of BV solutions here because we only need $L^\infty$ estimates in this paper and don't essentially use specific properties of BV solutions. The only thing we need to know is that given any $\alpha_0=\partial_r(r\phi)\big|_{C^+_{-1}}\in BV[0,+\infty)$, we will have a unique BV solution and if the total variation of $\alpha_0$ is sufficiently small, we will have a global BV solution in the sense that it is geodesically complete.  One can also impose that the BV data is more regular, such as $\alpha\in AC$ (which is absolutely continuous) or $\alpha\in C^1$ (both with finite total variation). The BV development will have correspondingly higher regularity. Note that $AC$ is a closed subspace of $BV[0,+\infty)$ and $C^1$ is dense in $AC$ (in BV topology). Apparently the initial data of $k$-self-similar solutions lie in AC class. 
\end{remark}

\begin{remark}
 A $k$-self-similar solution is not asymptotically flat because of self-similarity (for example, $m\to\infty$ as $r\to\infty$). The naked singularity solution is obtained by cutting off the data on $C_{-1}^+$ in sufficiently far region. In the following a $k$-self-similar solution of naked singularity refers to the asymptotically flat solution obtained by cutting off. 
\end{remark}

\begin{theorem}[Interior instability of $k$-self-similar naked singularities below the threshold]\label{thm:interior}
Given any $k^2\in (0,\frac{1}{3})$, a $k$-self-similar solution of naked singularity with initial data $\alpha_0$ on $C_{-1}^+$ and $p\in(0,k^2)$, the solution is unstable to black hole formation under $C^\frac{p}{1-p}$ perturbations \emph{totally supported in the original interior region}.

Precisely speaking, there exists a family of initial data $\alpha_{0,t}$ for sufficiently small $|t|\ne0$, $\alpha_{0,t}=\alpha_0$ in the exterior (relative to the development of $\alpha_0$), the maximal development of $\alpha_{0,t}$ for $t\ne0$ has a closed trapped surface and a complete future null infinity, and $\alpha_{0,t}\to \alpha_0$ in $C^\frac{p}{1-p}[0,+\infty)$ as $t\to0$. 
\end{theorem}

The word ``original'' refers to the fact that the interior region will change after nonlinear interior perturbations constructed in Theorem \ref{thm:interior},  which is significantly different from exterior perturbations which will not change the interior and exterior. The new singularity appears before the original one. But one may expect that the interior and exterior will not change under high regularity interior perturbations (consistent with the linear analysis in \cite{JS2}).  

\begin{remark}
It will be clear in the proof that the perturbations look like a family of bumps. The size of them shrinks to zero, and at the same time they approach the intersection of $C^{+}_{-1}$ and the past null cone of the original singularity. The convergence is then also in AC, that is, for any $t$, $\alpha_{0,t}\in AC$ and $\alpha_{0,t}\to \alpha_0$ in BV topology (total variation).  We can further make a cut off to make $\alpha_{0,t}\in C^1$ and the convergence is in $C^{\frac{p}{1-p}}$ on the whole half-line and $C^1$ away from one point. This can be made precise using the function spaces as introduced in \cite{JS2}.
\end{remark}

\begin{remark} One may think that interior perturbations leading to black hole formation can be obtained by first doing exterior perturbations leading to black hole and then do interior perturbations so that the trapped surfaces are still around there by Cauchy stability. But the interior perturbations constructed in Theorem \ref{thm:interior} is totally supported in the original interior region, and moreover a closed trapped surface forms there. This kind of interior perturbations cannot be obtained by a perturbative argument. We give an essentially new way to construct interior perturbations.
 \end{remark}
 
 In the  course of the proof, we prove a trapped surface formation theorem which applies to both interior and exterior regions, so the nonlinear exterior instability in all regularities below the threshold can also be obtained, which was also a conjecture in view of the exterior linear analysis (see \cite{JS1}). The statement is as follows.
 \begin{theorem}[Exterior instability of $k$-self-similar naked singularities below the threshold]\label{thm:exterior}
For any $p\in(0,k^2)$, the $k$-self-similar solutions of naked singularities are unstable under $C^\frac{p}{1-p}$  exterior perturbations, so that a black hole forms preceding the singularity. 
\end{theorem}
\begin{remark}
We can show the exterior H\"older instability for $\frac{p}{1-p}<\frac{2k^2}{2k^2+3}$ (smaller than $\frac{k^2}{1-k^2}$), using the argument in the last section of \cite{L-Z3}, in which we showed instability for naked singularity solutions of spherically symmetric Einstein--Euler system constructed in \cite{O-P, G-H-J} under H\"older  perturbations of an external scalar field.
\end{remark}

 \subsection{The proof strategy} As mentioned above, the instability argument in \cite{Chr99} relies on the infiniteness of \eqref{blueshift}, or \eqref{blueshift2} in \cite{Liu-Li, Li-Liu1, A}, along the past null cone of the singularity. The advantage of exterior perturbations is that they will not affect the singularity, and the main difficulty in constructing interior perturbations is that interior perturbations may change the location of the singularity and one doesn't know a priori what would happen after perturbations. 
 
For simplicity, in constructing interior perturbations, one would first consider perturbations supported away from the center. In this case, we can still have an incoming cone (say $\mathcal{N}_1$, see Figure \ref{fig:nakedinterior}), intersecting the initial outgoing cone at the ``starting point'' of the perturbation, with the property that its past will not affected under such a perturbation. This problem can then be still viewed as an exterior problem to $\mathcal{N}_1$. Although one would have trapped surface formation if the perturbation is large, when such a perturbation becomes smaller, it cannot be expected that trapped surface will form, because the center of $\mathcal{N}_1$ is regular and there is no infinite blueshift helping the trapped surface formation. 

Nevertheless, one would expect that a ``finite'' blueshift can still contribute to trapped surface formation, which can be seen implicitly in the our previous work \cite{Li-Liu1} (Case 2 of Theorem 6.1), in which most theorems are applicable to the most general case, not only  that \eqref{blueshift2} is infinite. Our main new strategy is to push $\mathcal{N}_1$ to $\mathcal{N}$ (see Figure \ref{fig:nakedinterior})  and establish trapped surface formation theorem for each $\mathcal{N}_1$ before $\mathcal{N}$ in a very precise way. Observing that the blueshift tends to infinity, the perturbations leading to trapped surface formation can be made arbitrarily small, with the ``starting points'' of the perturbations approaching $\mathcal{N}$. Once the trapped surface formation has been established, the black hole will eventually form,  as showed by Dafermos \cite{D}, saying that even  a single trapped surface is enough to guarantee that the maximal solution has a complete future null infinity and an event horizon,  in many spherically symmetric models including Einstein--scalar field. Moreover, for $k$-self-similar naked singularities, the rate of \eqref{blueshift2} becoming infinity is particularly clear, so we are able to show interior instability in all regularities below the threshold by refining the original arguments in \cite{Liu-Li, L-Z3}.

 \begin{figure}
\includegraphics[width=3 in]{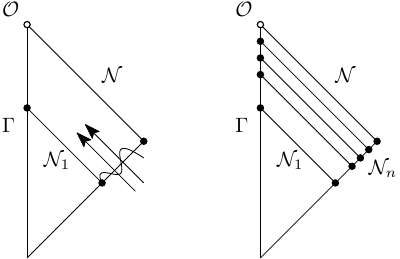}
\caption{Interior perturbations}
\label{fig:nakedinterior}
\end{figure}

 \subsection{The main results for general naked singularities and their implications} An idea immediately coming to mind is to generalize the above strategy to general naked singularities, with the property that the integrals (blueshifts) \eqref{blueshift} or \eqref{blueshift2} along a family of incoming null cones $\mathcal{N}_n$ tends to infinity when $\mathcal{N}_n$ approach the past null cone $\mathcal{N}$ of the singularity $\mathcal{O}$.  The precise setup is the following, where we use the notations in \cite{Chr93}.  Suppose that we have a spherically symmetric BV solution of the Einstein--scalar field system  and the quotient manifold $\mathcal{Q}$ is equipped with double null coordinates $(u,v)$ so that $u$ is constant along outgoing null cones from a point at the center $\Gamma$, and $v$ is constant along conjugate incoming null cones.  
 Suppose that the center $\Gamma$ is given by $u=v$ and the ``singularity'' $\mathcal{O}$ corresponds to $u=v=0$. The past null cone of $\mathcal{O}$ would then be $v=0$. We only consider the part of the future of some outgoing cone $C_{-1}^+: u=-1$, on which $v$ is chosen to be $\frac{1}{2}r$ and we assume the part $(u,u), u\in[-1,0)$ of the center is regular (in the sense that $\mu\to0$ along each $C^-_v$ for $v<0$). The point $\mathcal{O}$ is supposed to be singular in the following sense:  there exists a sequence of $v_n$ approaching $0$ monotonically, with the integrals
 \begin{equation}\label{blueshfit4}\int_{\mathcal{N}_n}r\left(\frac{n\phi}{nr}\right)^2\D r\to\infty, n\to\infty,\end{equation}
where $\mathcal{N}_n$ is simply the level set $v=v_n$ (see Figure \ref{fig:nakedinterior} again). $v_n$ approaching $0$ means that $\mathcal{N}_n$ approaches $\mathcal{N}$, the past null cone of $\mathcal{O}$. Then we will prove that such singular spherically symmetric solutions are unstable to black hole formation under BV perturbations from the initial data on $C_{-1}^+$. However, note that a singular solution in the above sense need not to be a naked singularity solution,  which is similar to Christodoulou's proof of instability in \cite{Chr99}, where he only assumed that \eqref{blueshift} is infinite on the past null cone $\mathcal{N}$ of $\mathcal{O}$ and whether the exterior extension is a naked singularity is not important.  The precise statement of the above instability is the following.
 \begin{theorem}[Interior BV instability of general naked singularities]\label{thm:BVinterior}
 Suppose that the initial data  $\alpha_0$ given on $C_{-1}^+$ leads to a singular solution in the above sense. Then the maximal future development has a closed trapped surface and a complete future null infinity, or there exists a sequence of initial data $\alpha_{0,n}, n\in\mathbb{N}$ that differ from $\alpha_0$ in the interior region (relative to the development of $\alpha_0$, that is, $v<0$),  so that the maximal future development of $\alpha_{0,n}$ has a closed trapped surface and a complete future null infinity, and $\alpha_{0,n}-\alpha_0$ is absolutely continuous and tends to zero as $n\to\infty$ in $BV[0,+\infty)$.
 \end{theorem}

The basic strategy of the proof is similar to the proof for $k$-self-similar solutions. The main difference is that in general cases we do not have precise information about the solution other than  \eqref{blueshfit4}, so we should establish the trapped surface formation theorem including all possible cases. For the same reason,  we cannot expect the perturbations can be made small in a more regular norm than BV, and we will not pursuit that the perturbations are supported totally in the original interior region.

 Interestingly, the above theorem provides a new insight into Christodoulou's program in proving the weak cosmic censorship for Einstein--scalar field system. As in the case of a single null cone, where the infiniteness of  \eqref{blueshift} and \eqref{blueshift2} are related, the condition \eqref{blueshfit4}, based on which Theorem \ref{thm:BVinterior} is established,  is also a consequence of 
  \begin{equation}\label{blueshift3}\int_{\mathcal{N}_n}\frac{1}{r}\frac{\mu}{1-\mu}\D r\to\infty, n\to\infty.\end{equation}
 This can be deduced in a similar way as we did in \cite{Liu-Li} on a single null cone. In fact, when the center is regular ($\mu\to0$ as approaching the center along each $\mathcal{N}_n$), the infiniteness  of both integrals  are equivalent (see Remark \ref{blueshiftequal} below). Now, if we assume the infiniteness of \eqref{blueshift2} (or \eqref{blueshift}) on $\mathcal{N}$ as in the works \cite{Chr99, Liu-Li, Li-Liu1}, then by Fatou's Lemma,   \eqref{blueshfit4} (or \eqref{blueshift3}) holds for any sequences $\mathcal{N}_n$ approaching $\mathcal{N}$. Thus Theorem \ref{thm:BVinterior} is applicable, and moreover the perturbed family can be made to a one-parameter instead of a discrete family.

 Beside this, there is a more direct way to justify our assumption \eqref{blueshfit4} (or \eqref{blueshift3})) in applying Theorem \ref{thm:BVinterior}.  In fact, the condition \eqref{blueshift3} can directly be proven to be a blow up criterion of $C^1$ spherically symmetric solutions\footnote{$C^1$ solutions roughly mean that the wave function $\phi$ and the area radius $r$ are $C^2$ functions in $(u,v)$ including the center, and certain regularity conditions hold on the center. It arises from $C^1$ initial data. See also \cite{Chr93}.} of the Einstein--scalar field system. Its proof can be found in the Appendix B of \cite{A-T}, in which An-Tan studied the weak cosmic censorship conjecture of spherically symmetric Einstein--Maxwell--(complex) scalar field system, and they introduced (a charged version of) this integral to establish an extension principle. In the Einstein--scalar field case, roughly speaking, if the $C^1$ solution exists in $v<0$ and the integrals
 \begin{equation}\label{blueshiftbounded}\int\frac{1}{r}\frac{\mu}{1-\mu}\D r\end{equation}
over the incoming null cone $v=v_0$ are uniformly bounded for all $v_0<0$, then the $C^1$ solution extends across $\mathcal{O}$ and then $\mathcal{O}$ is not a singularity. Christodoulou's extension principle in \cite{Chr93} is that $\mu$ is sufficiently small when approaching $\mathcal{O}$ from its causal past.  In fact, by changing the variable $r$ to $u$, the integrant $\frac{-\partial_ur}{r}\frac{\mu}{1-\mu}$ appears as the coefficient of the linear term of a differential equation for second order transversal derivatives of $\phi$. Either $\mu$ being small or the integral \eqref{blueshiftbounded} being uniformly bounded is sufficient for absorbing this linear term.

 Theorem \ref{thm:BVinterior} can then be applied to show \emph{instability of AC naked singularity solutions} without utilizing exterior instability. The argument is as follows:  Let $\mathcal{A}$ and $C^1$ be space of all AC and $C^1$ initial data with bounded total variation, $\mathcal{R}\subset\mathcal{A}$ be the space of data leading to complete maximal future developments, and $\mathcal{S}=\mathcal{A}-\mathcal{R}$. Let $\mathcal{G}\subset\mathcal{S}$ be the space of initial data leading to maximal developments possessing a closed trapped surface and a complete future null infinity, and $\mathcal{E}=\mathcal{S}-\mathcal{G}$.  By definition, $\mathcal{E}$ includes data leading to naked singularity solutions. Choose $\alpha_0\in\mathcal{E}$, note that $C^1$ is dense in $\mathcal{A}$, we can choose $\alpha_{0,m}\in C^1$ so that $\alpha_{0,m}\to\alpha_0$ in AC. If (by possibly passing to a subsequence)  $\alpha_{0,m}\in C^1\cap \mathcal{S}$, then the blow up criterion \eqref{blueshift3} holds for every $m$ before their respective first singularities. By Theorem \ref{thm:BVinterior}, we can perturb each $\alpha_{0,m}$ a little bit so that  $\alpha_{0,m}\in \mathcal{G}$ and still converges to $\alpha_0$ in AC.  It could happen that such $\alpha_{0,m}\in C^1$ are all in $\mathcal{R}$, that is, lead to complete solution. In any case, we have shown the instability:  $\mathcal{A}-\mathcal{E}=\mathcal{R}\cup\mathcal{G}$ is dense in $\mathcal{A}$.

\begin{remark}
The denseness of $\mathcal{A}-\mathcal{E}$ alone does not mean it is generic in $\mathcal{A}$. In Christodoulou's work \cite{Chr99}, the exterior perturbations constructed there have the form of a line $\alpha_0+tf, t\in\mathbb{R}$ for some $f$, and such different   lines do not intersect.  This means, in addition to the denseness of $\mathcal{A}-\mathcal{E}$,  that $\mathcal{E}$ has codimension at least $1$ in $\mathcal{A}$, which is the genericity in the statement of the weak cosmic censorship conjecture. But based on \eqref{blueshift3}, we can only construct interior perturbations for discrete $n$. Moreover, we also don't know whether different families of perturbations intersect. Therefore, if we bring back the exterior instability, it is not clear whether we can have higher codimension instability by adding interior perturbations. 

Nevertheless, note that $\mathcal{A}$ is a complete metric space, one can still draw conclusion on a weaker genericity in the sense of Baire category\footnote{The genericity established for non-spherically symmetric gravitational perturbations in \cite{Li-Liu1} is also in the sense of Baire category.} by Cauchy stability in BV topology. First, Cauchy stability implies that $\mathcal{G}$ is open in $\mathcal{A}$ (the existence of a closed trapped surface is an open condition, and \cite{D} applies to conclude the existence of the complete future null infinity). Second, let $\mathcal{R}_T\subset\mathcal{A}$ be the data leading to a maximal future development with the property that the first singularity does not appear before proper time $T>0$ measured from the initial vertex. Then $\mathcal{R}=\displaystyle\bigcap_{N=1}^\infty \mathcal{R}_N$ and Cauchy stability implies that $\mathcal{R}_N$ is open for each $N$. Therefore
$$\mathcal{A}-\mathcal{E}=\mathcal{R}\cup\mathcal{G}=\bigcap_{N=1}^\infty(\mathcal{R}_N\cup\mathcal{G})$$
is a dense set and a countably intersection of open sets, and hence a complement of a first category set, which can be viewed as a generic set.

\end{remark}

\subsection*{Acknowledgement}

 The author would like to thank Tingting Li for carefully reading the manuscript. This work is supported by National Key R\&D Program of China (No. 2022YFA1005400) and NSFC (12326602, 12141106).

\section{The apriori estimates}\label{section:estimate} 

In this section we prove an a priori estimates, based on which all trapped surface formation theorems are derived.  It is a refinement of the a priori estimates derived in our previous works \cite{Liu-Li}.

\subsection{Double null coordinates and equations}

We will follow closely the notations in \cite{Liu-Li}. We use double null coordinates $(\ub,u)$, where $\ub$, $u$ are optical functions, with their level sets $\Cb_{\ub}$ and $C_u$ being incoming and outgoing null cones invariant under the $SO(3)$ action respectively. In the quotient spacetime, $\Cb_{\ub}$ and $C_u$ are simply incoming and outgoing null rays. We denote
\begin{align*}
L=\frac{\partial}{\partial \ub},\ \Lb=\frac{\partial}{\partial u},
\end{align*}
and  the lapse function $\Omega$ by
\begin{align*}
-2\Omega^2=g(L,\Lb).
\end{align*}
The metric then has the form
\begin{align*}
-2\Omega^2(\D\ub\otimes\D u+\D u\otimes\D\ub)+r^2\D\sigma_{\mathbf{S}^2}
\end{align*}
where the area radius function $r=r(\ub,u)$ is defined by
\begin{align*}
\text{Area}(S_{\ub,u})=4\pi r^2,
\end{align*}
and $\D\sigma_{\mathbf{S}^2}$ is the standard metric of the unit sphere.

 The unknowns of the Einstein-scalar field equations are then $r$, $\Omega$ and the scalar field function $\phi$. We define the null expansions relative to the normalized pair of null vectors $\Omega^{-2}L$, $\Lb$ and the mass function $m$ by
\begin{align*}
h=\Omega^{-2}D r,\ \hb=\Db r,\ m=\frac{r}{2}(1+h\hb),
\end{align*}
where $D$ and $\Db$ are the restrictions on the orbit spheres of the Lie derivatives along $L$ and $\Lb$. In our case, $D$ and $\Db$ are simply the ordinary derivatives $\partial_{\ub}$ and $\partial_u$. The $D$ and $\Db$ derivative of the lapse $\Omega$\begin{align*}
 \omega=D\log\Omega, \ \omegab=\Db\log\Omega.
\end{align*}
But $\omegab$   is not needed in this paper. Finally, we denote
\begin{align*}
L\phi=\frac{\partial}{\partial\ub}\phi,\ \Lb\phi=\frac{\partial}{\partial u}\phi.
\end{align*}

The Einstein--scalar field system \eqref{ES} implies that following null structure equations (also see \cite{Liu-Li}):
\begin{align}
\label{Dh}Dh=&-r\Omega^{-2}(L\phi)^2,\\
\label{Dbh}\Db(\Omega^2h)=&-\frac{\Omega^2(1+h\hb)}{r},\\
\label{Dhb}D\hb=&-\frac{\Omega^2(1+h\hb)}{r},\\
\label{Dbhb}\Db(\Omega^{-2}\hb)=&-r\Omega^{-2}(\Lb\phi)^2,\\
\label{Dbomega}\Db\omega=&\frac{\Omega^2(1+h\hb)}{r^2}-L\phi\Lb\phi.
\end{align}
And the wave equation reads in the double null coordinates:
\begin{align}
\label{DbLphi}\Db(rL\phi)=&-\Omega^2h\Lb\phi,\\
\label{DLbphi}D(r\Lb\phi)=&-\hb L\phi.
\end{align}
These two equations are in fact the same equation. 

\subsection{The a priori estimates}

 We consider a double null initial value problem with initial data given on $\Cb_0$ where $\ub=0$ and $C_{u_0}$ where $u=u_0$ for some $u_0<0$.  On $\Cb_0$, we have a freedom to choose the function $u$, and we set $u=-r$. We denote the restrictions on $\Cb_0$ of some geometric quantities and derivatives of $\phi$, which are considered as functions of $u$:
$$\psi=\psi(u)=r\Lb\phi\Big|_{\Cb_0},\ \varphi=\varphi(u)=rL\phi\Big|_{\Cb_0},\ \Omega_0=\Omega_0(u)=\Omega\Big|_{\Cb_0},\ h_0=h_0(u)=h\big|_{\Cb_0}.$$
Since $u=-r$ on $\Cb_0$, we must have $\hb\big|_{\Cb_0}\equiv-1$.

We also have a freedom to choose $\Omega$ along $C_{u_0}$, that is, the freedom of choosing $\ub$ on $C_{u_0}$. We leave this freedom by only setting $\Omega(0,u_0)\le 1$ for convenience. On $\Cb_0$, equation \eqref{Dbhb} reads 
\begin{align*}
\frac{\partial}{\partial u}\log\Omega_0=-\frac{1}{2}\frac{\psi^2}{|u|},
\end{align*} 
and hence
\begin{align}\label{Omega_0}
 \log\frac{\Omega_0^2(u)}{\Omega_0^2(u_0)}=-\int_{u_0}^u\frac{\psi^2(u')}{|u'|}\D u'.
 \end{align}
 Therefore, $\Omega$ along $\Cb_0$ is determined. In particular, $\Omega_0$ is decreasing, and we have $\Omega_0(u)\le1$ for all $u\in[u_0,0)$. 
 
 \begin{remark}\label{blueshiftequal}
We will show that the conditions \eqref{blueshfit4} and \eqref{blueshift3} are equivalent assuming the center is regular. The above mentioned double null coordinates can be constructed for $\Cb_0=\mathcal{N}_n$ for each $n$.  From \eqref{Dbh}, we have
\begin{align*}
\frac{\partial}{\partial u}(\Omega_0^2h_0)=-\frac{\Omega_0^2(1-h_0)}{|u|},
\end{align*}
and hence
\begin{align}\label{Omega_0^2h_0}
 -\log\frac{\Omega_0^2(u)h_0(u)}{\Omega_0^2(u_0)h_0(u_0)}=\int_{u_0}^u\frac{1}{|u'|}\left(\frac{1}{h_0(u')}-1\right)\D u'=\int_{u_0}^u\frac{1}{|u'|}\frac{\mu(u')}{1-\mu(u')}\D u'.
\end{align}
Comparing \eqref{Omega_0} and  \eqref{Omega_0^2h_0}, we have
$$\int_{u_0}^u\frac{1}{|u'|}\frac{\mu(u')}{1-\mu(u')}\D u'=\int_{u_0}^u\frac{\psi^2(u')}{|u'|}\D u'-\log\frac{h_0(u)}{h_0(u_0)}.$$
Since $\mu_0=1-h_0$, regularity on the center means that $h_0(u)\to1$ as $u\to0$, and hence \eqref{blueshfit4} and \eqref{blueshift3} are equivalent.
\end{remark}

In $k$-self-similar cases, we need to utilize the upper bound of $\Omega_0$ to do more precise estimates. Let $u_1\in (u_0,0)$ and $p>0$ be such that  for any $u\in [u_0,u_1]$,
\begin{equation*} 
 \frac{\Omega_0^2(u)}{\Omega_0^2(u_0)}\le\left(\frac{|u|}{|u_0|}\right)^{p}.
\end{equation*}
Because we have set $\Omega_0(u_0)\le 1$ and for simplicity we assume $u_0\le -1$, then this condition becomes
\begin{equation}\label{Omegaupper}
\Omega_0^2(u)\le|u|^{p}.
\end{equation}

\begin{remark}The above assumptions do not imply that the vertex of $\Cb_0$ is singular since we do not assume anything for $u\in(u_1,0)$.
 On the other hand,   the past null cone of the singularity in the $k$-self-similar solutions (see Section \ref{section:kselfsimilar}) do verify the above assumptions for $p=k^2$ and arbitrarily $u_1\in (u_0,0)$. 
 \end{remark}
 To each $p$, we define a positive function $w=w_p(u)$ such that
\begin{equation}\label{defw} w_p(u)=\max\left\{1, |u|^{1-2p}\int_{u_0}^{u}|u'|^{2p-2}|\psi|\D u'\right\}\end{equation}
which is used to estimate $\psi$. We also define $w_1=\displaystyle\sup_{u_0\le u\le u_1}w(u)$. It will be clear later that in $k$-self-similar cases $\psi$ is bounded by a constant and so does $w_1$. This function is introduced to handle general cases.

  We will prove the following a priori estimates.
\begin{theorem}\label{estimate}
Suppose that $\ub_1>0, u_0, u_1, p$ with $ u_0\le -1$, $u_0\le u_1<0$ and $p\in [0,\frac{1}{3})$ are chosen and \eqref{Omegaupper} is satisfied on $\Cb_0$. We also assume that $0\le h_0\le 1$ on $\Cb_0$\footnote{This condition means that there are no trapped surfaces are on $\Cb_0$ and the mass function $m$ is nonnegative.}. Then there exists a universal\footnote{The a priori estimates can be easily generalized to $p\in[0,\frac{1}{2})$ and $\varepsilon$ depends on the gap between $p$ and $\frac{1}{2}$.} $\varepsilon>0$ such that the following statement is true. Let $a=a(\ub_1, u_0,u_1)$ be satisfied
\begin{align}\label{defa}
\max\left\{1,\sup_{0\le\ub\le\ub_1}\left( \ub^{-\frac{p}{1-p}}|rL\phi(\ub,u_0)-\varphi(u_0)|,|u_0|^{1-p}|\omega(\ub,u_0)|\right) ,\sup_{u_0\le u\le u_1}|u|^{-p}|\varphi(u)|\right\}\le a
\end{align}
Then in the region $(\ub,u)\in[0,\ub_1]\times[u_0,u_1]$ where
\begin{align}\label{smallness}
\ub_1|u_1|^{p-1}w_1^2a\le\varepsilon
\end{align}
we have the following estimates:
\begin{align}
\label{estimate-bigomega}\frac{1}{2}\Omega_0\le \Omega\le &2\Omega_0,\\
\label{estimate-r} \frac{1}{2}|u|\le r\le& 2|u|,
\end{align}
and\footnote{The notation $A\lesssim B$ means $A\le cB$ for some  universal constant $c$.}
\begin{align}
\label{estimate-Lphi}|rL\phi|\lesssim& |u|^{p}a,\\
\label{estimate-Lbphit}|r\Lb\phi-\psi|\lesssim&\ub|u|^{p-1}a,\\
\label{estimate-ht}|h-h_0|\lesssim&\ub|u|^{p-1}(|u|^{-p}\Omega_0^{2})^{-1}a^2,\\
\label{estimate-hbt}|\hb+1|\lesssim&\ub|u|^{p-1}a,\\
\label{estimate-omega}|u||\omega|\lesssim&|u|^{p}wa.
\end{align}
\end{theorem}

\begin{proof}
We follow the line of the proof in \cite{Liu-Li} with modifications due to the introduction of $p$. Let us start with the bootstrap assumptions:
\begin{align}
\label{bootstrapLphi}|rL\phi|\lesssim&|u|^{p}a\varepsilon^{-\delta},\\
\label{bootstrapomega}|u||\omega|\lesssim&|u|^{p}wa\varepsilon^{-\delta},\\
\label{bootstrapLbphit}|r\Lb\phi-\psi|\lesssim&\ub|u|^{p-1}a\varepsilon^{-\delta},\\
\label{bootstrapht}|h-h_0|\lesssim&\ub|u|^{p-1}(|u|^{-p}\Omega_0^{2})^{-1}a^2\varepsilon^{-\delta},\\
\label{bootstraphbt}|\hb+1|\lesssim&\ub|u|^{p-1}a\varepsilon^{-\delta},
\end{align}
where $\delta>0$ is a small constant ($\delta=\frac{1}{4}$ is sufficient). From \eqref{bootstrapomega}, we have
\begin{align*}
|\log\Omega-\log\Omega_0|\le\int_0^{\ub}|\omega|\D\ub'\lesssim \ub|u|^{p-1}wa\varepsilon^{-\delta}.
\end{align*}
By choosing $\varepsilon>0$ sufficiently small and $\ub|u|^{p-1}wa\le\varepsilon$, we have
\begin{align*}
|\log\Omega-\log\Omega_0|\le \log 2
\end{align*}
and therefore \eqref{estimate-bigomega} holds. Moreover, since $D\Omega=\Omega\omega$, we have
\begin{align}\label{estimate-Omega-Omega0}
|\Omega-\Omega_0|\le\int_0^{\ub}\left|\Omega\omega\right|\D\ub'\lesssim \ub|u|^{p-1}w\Omega_0a\varepsilon^{-\delta}.
\end{align}

For $r$, we note that, from \eqref{bootstrapht} and $0\le h_0\le 1$,
\begin{align}\label{estimate-h}
|\Omega^2h|\lesssim\Omega_0^2\left(h_0+\ub|u|^{p-1}(|u|^{-p}\Omega_0^{2})^{-1}a^2\varepsilon^{-\delta}\right)\lesssim |u|^p a,
\end{align}
 Here we use the fact that $h_0\le 1\le a$, $\Omega_0^2\le |u|^p$ and $\varepsilon$ sufficiently small.  We then use the equation $Dr=\Omega^2h$ to obtain
\begin{align}\label{estimate-rt}
|r-|u||\le\int_0^{\ub}|\Omega^2h|\D\ub\lesssim \ub|u|^p a,
\end{align}
We then deduce that $|r-|u||\lesssim \varepsilon|u|$ and \eqref{estimate-r} holds for $r$ if $\varepsilon$ is sufficiently small.

The most delicate part is $L\phi$. We consider the equation \eqref{DbLphi}. We write
\begin{align}\label{DbrLphi-varphi}
\frac{\partial}{\partial u}(rL\phi-\varphi)=-\left(\Omega^2hr^{-1}(r\Lb\phi)-\Omega_0^2h_0|u|^{-1}\psi\right).
\end{align}
Using \eqref{bootstrapLbphit}, \eqref{bootstrapht}, \eqref{estimate-Omega-Omega0}, \eqref{estimate-h}, \eqref{estimate-rt}, the right hand side can be estimated by
\begin{align*}
&|\Omega^2hr^{-1}(r\Lb\phi)-\Omega_0^2h_0|u|^{-1}\psi|\\
\lesssim&|\Omega^2-\Omega_0^2||h_0|u|^{-1}\psi|+|\Omega^2||h-h_0|||u|^{-1}\psi|+|\Omega^2h||r^{-1}-|u|^{-1}||\psi|+|\Omega^2hr^{-1}||r\Lb\phi-\psi|\\
\lesssim&|u|^{-1}\varepsilon^{-\delta}\left(\ub|u|^{p-1}\Omega_0^2wa|\psi|+\ub|u|^{p-1}|u|^p a^2(1+|\psi|)\right).
\end{align*}
Using $\Omega_0^2\le |u|^p$, $a\ge 1$ again, and \eqref{defw}, we have
\begin{align}\label{proof-estimate-Lphiimp}
\int_{u_0}^u|\Omega^2hr^{-1}(r\Lb\phi)-\Omega_0^2h_0|u'|^{-1}\psi|\D u'\lesssim \varepsilon^{-\delta} \ub|u|^{p-1}|u|^p w_1^2a^2\end{align}

Integrating the equation \eqref{DbrLphi-varphi} and using this estimate, we then have
\begin{equation}\label{estimate-Lphit}
\begin{split}
|rL\phi-\varphi|
\lesssim|rL\phi-\varphi|\big|_{C_{u_0}}+ \varepsilon^{-\delta} \ub|u|^{p-1}|u|^p w_1^2a^2.
\end{split}
\end{equation}
By the definition of $a$ in \eqref{defa}, 
$$|rL\phi-\varphi|\big|_{C_{u_0}}\le a\ub^{\frac{p}{1-p}}\le |u|^{p}a.$$
 Then the estimate \eqref{estimate-Lphi} follows, improving the bootstrap assumption \eqref{bootstrapLphi}.

For $\Lb\phi$, we note that, from \eqref{bootstraphbt}, and if $\varepsilon$ is sufficiently small, 
\begin{align}\label{estimate-hb}
|\hb|\lesssim1+\ub|u|^{p-1}a\varepsilon^{-\delta}\lesssim1, 
\end{align}
We then integrate the equation \eqref{DLbphi} and obtain, by \eqref{estimate-Lphi} and \eqref{estimate-hb} we have proved above,
\begin{align*}
|r\Lb\phi-\psi|\lesssim\int_0^{\ub}|\hb L\phi|\D\ub\lesssim\ub|u|^{p-1}a,
\end{align*}
which is \eqref{estimate-Lbphit}, improving \eqref{bootstrapLbphit}. We also have
\begin{equation}\label{estimate-Lbphi}
|r\Lb\phi|\lesssim|\psi|+\ub|u|^{p-1}a.
\end{equation}

For $h$ and $\hb$, we use the equations \eqref{Dh} and \eqref{Dhb}. Using \eqref{estimate-Lphi}, we have
\begin{align*}
|h-h_0|\lesssim\int_0^{\ub} |r\Omega^{-2}(L\phi)^2|\D \ub\lesssim \Omega_0^{-2}\ub|u|^{2p-1}a^2,
\end{align*}
which is the desired estimate \eqref{estimate-ht}, improving \eqref{bootstrapht}. From \eqref{Dhb}, \eqref{estimate-h} and \eqref{estimate-hb}, we have
\begin{align*}
|\hb+1|\lesssim\int_0^{\ub}\left|\frac{\Omega^2(1+h\hb)}{r}\right|\D\ub\lesssim \ub|u|^{p-1}a, 
\end{align*}
which is the desired estimate \eqref{estimate-hbt}, improving \eqref{bootstraphbt}.

For $\omega$, we use the equation \eqref{Dbomega}. The right hand side of \eqref{Dbomega} can be estimated by, using   \eqref{estimate-Lphi}, \eqref{estimate-Lbphi}, \eqref{estimate-h} and \eqref{estimate-hb},
\begin{align*}
\left|\frac{\Omega^2(1+h\hb)}{r^2}-L\phi\Lb\phi\right|\lesssim|u|^{-2}|u|^p (1+|\psi|)a.\end{align*}
Integrating \eqref{Dbomega} and using \eqref{Omega_0}, we then have
\begin{align*}
|\omega|\lesssim|u|^{p-1}wa,
\end{align*}
which is the desired estimate \eqref{estimate-omega}, improving \eqref{bootstrapomega}. All the bootstrap assumptions \eqref{bootstrapLphi}-\eqref{bootstraphbt} are improved and the proof is then completed.

\end{proof}

\section{Instability of the $k$-self-similar solution}\label{section:kselfsimilar}

\subsection{The $k$-self-similar  solution revisit}\label{sec:revisit}

We review some more details about the $k$-self-similar solutions needed in our proof. Recall that in \cite{Chr94} it was studied in self-similar Bondi coordinates \eqref{Bondimetric}.  To transfer the system \eqref{ES} and \eqref{Wave} to an autonomous system, Christodoulou introduced
\begin{equation}\label{defthetaalpha}\theta=r\frac{\partial\phi}{\partial r},\ \alpha=\beta^{-1}=\left(1-\frac{\mathrm{e}^{\nu-\lambda}}{-\frac{2r}{u}}\right)^{-1},\end{equation}
and a new self-similar parameter
$$s=\log\left(-\frac{r}{u}\right),$$
then from the Einstein equations \eqref{ES}, he deduced that
\begin{equation}\label{autonomous}
\begin{split}
\frac{\D\alpha}{\D s}&=\alpha((\theta+k)^2+(1-k^2)(1-\alpha)),\\
\frac{\D\theta}{\D s}&=k\alpha(k\theta-1)+\theta((\theta+k)^2-(1+k^2)).
\end{split}
\end{equation}
 Note that $\frac{\partial}{\partial r}$ is outgoing null according to \eqref{Bondimetric} and then $\theta$ is $r$ times the outgoing null derivative of $\phi$, and moreover this derivative is parametrized by $r$. Conversely, once $\alpha$ and $\theta$ are solved, $\lambda$ can be determined by
\begin{equation}\label{solvelambda}\mathrm{e}^{2\lambda}=1+k^2+\frac{\beta}{1-\beta}(\theta+k)^2\end{equation}
and then $\nu$ and $\phi$ can be solved in view of \eqref{defthetaalpha}.

The center corresponds to $s=-\infty (r=0)$, regularity on $\Gamma$ requires $\lambda\to0$ as $s\to-\infty$ and we impose  $\nu\to0$ by rescaling $u$ on $\Gamma$ (such that $u$ is the proper time along $\Gamma$). From  the definition of $\alpha$ we will have\footnote{Here $a\sim b$ means $\lim\frac{a}{b}=1$. }
\begin{equation*}\alpha\sim-2\mathrm{e}^s, s\to -\infty.\end{equation*}
By \eqref{solvelambda} we will have $\theta\to 0$ or $-2k$ but the latter possibility can be ruled out for $(\alpha,\theta)=(0,2k)$ is not a critical point of \eqref{autonomous}. Further analysis of \eqref{autonomous} around $s=-\infty$ suggests that
\begin{equation*}\theta\sim k\mathrm{e}^s, s\to -\infty.\end{equation*}
Christodoulou showed that when $0<k^2<1$, the solution can be solved from $s=-\infty$ to some $s=s_*<+\infty$, where $\alpha\to-\infty$, corresponding to the past null cone of $\mathcal{O}$. The solution in $s\in (-\infty, s_*)$ then corresponds to the interior solution. Further analysis showed that $\alpha<0$ in the interior region, $\alpha$ has the asymptotic behavior
\begin{equation}\label{alphas*}\alpha\sim-\frac{1}{1-k^2}\frac{1}{s_*-s}, s\to s_*.\end{equation}
and $\theta$ has the asymptotic behavior 
\begin{equation}\label{thetas*}\theta\to\frac{1}{k}, s\to s_*.\end{equation}

In view of the discussions before the statement of Theorem \ref{estimate}, the most important quantity on $\Cb_0$ is $\psi$, $r$ times the incoming null derivative (parametrized by $r$) of $\phi$. In terms of self-similar Bondi coordinates, it is the opposite of the following quantity
\begin{equation}\label{defzeta}\zeta=-2r\mathrm{e}^{\lambda-\nu}\frac{\partial\phi}{\partial u}+r\frac{\partial\phi}{\partial r}.\end{equation}    
By $k$-self-similarity, $\theta$ and $\zeta$ should be related. In fact we have (also see \cite{Chr94})
$$\zeta=-\frac{\theta+k\alpha}{\alpha-1},$$
and then by \eqref{thetas*}, 
\begin{equation}\label{zetas*} \zeta\to -k, s\to s_*.\end{equation}

When $0<k^2<\frac{1}{3}$, the solutions extend  to $s_*$ only in a finite H\"older regularity\footnote{Here $a\approx b$ means $\frac{b}{c}\le a\le cb$ for some $c>0$ when $s$ is sufficiently close to $s_*$.}:
\begin{align*}\beta=\frac{1}{\alpha}\sim& -(1-k)^2(s_*-s),\\
\theta-\frac{1}{k}\approx& (s_*-s)^{\frac{k^2}{1-k^2}},
\end{align*}
which cannot be eliminated by coordinate transformation. There exists (also finite regularity) exterior extensions from $s=s_*$ to $s=+\infty$, yielding naked singularity solutions.

\subsection{A trapped surface formation theorem}

We are going to prove a sharp trapped surface formation theorem adapted to the $k$-self-similar background. 

\begin{theorem}\label{fots} Suppose in addition to Theorem \ref{estimate} that 
\begin{align}\label{Omegaupper2}
\Omega_0^2(u_1)\le  |u_1|^{\widetilde{p}}
\end{align}
where $\widetilde{p}\ge p$, and for any $u\in [u_0,u_1]$,
\begin{equation}\label{psivarphibound}|\psi(u)|\le 1, |(\Omega_0^2h_0)^{-1}\varphi(u)|\le K\end{equation}
for some $K\ge1$. Suppose also that
\begin{align}\label{lowerbound}\int_0^{\ub_1}\left|rL\phi\big|_{C_{u_0}}-\varphi(u_0)\right|^2\D\ub\ge c_1 |u_1|^{\widetilde{p}+1},\end{align}
for some $c_1$ depending on $K$, where $-1\le u_1<0$ and 
\begin{align}\label{deltau1}\ub_1 |u_1|^{p-1}a=\varepsilon|u_1|^{\frac{\widetilde{p}-p}{2}},\end{align}
where $\varepsilon$ may be smaller than in Theorem \ref{estimate}. Then $S_{\ub_1,u_1}$ is a closed trapped surface.
\end{theorem}
\begin{proof}
Note that the bound for $\psi$ leads to an upper bound for $w$ defined in \eqref{defw}:
$$w\le\frac{1}{1-2p} \le 3,$$
as $p\in(0,\frac{1}{3})$. Then \eqref{deltau1} implies \eqref{smallness} if we choose a smaller $\varepsilon$, and we can use the estimates established in Theorem \ref{estimate}.

By \eqref{DbrLphi-varphi}, we will consider
$$\frac{\partial}{\partial u}|rL\phi(\ub, u)-\varphi(u)|^2=-2\left(\Omega^2hr^{-1}(r\Lb\phi)-\Omega_0^2h_0|u|^{-1}\psi\right)(rL\phi(\ub, u)-\varphi(u)).$$
We need to revisit the estimate \eqref{proof-estimate-Lphiimp} to deal with the $a^2$ term coming from $h-h_0$ or $h$. Together with \eqref{estimate-Lphi} and $w\le 3$, we have  (the factor $\varepsilon^{-\delta}$ is absent now because the bootstrap assumptions there are now improved) 
$$\left|\frac{\partial}{\partial u}|rL\phi(\ub, u)-\varphi(u)|^2\right|\lesssim |u|^p a\cdot\left(\ub|u|^{p-2}\Omega_0^2a+|\Omega^2h|\ub|u|^{p-2}a+ \Omega_0^2|u|^{-1}|h-h_0| \right).$$
Using the equation \eqref{Dh} directly, we have, by \eqref{psivarphibound},
$$|h-h_0|\lesssim |u|^{-1}\Omega_0^{-2}\int_0^{\ub}(rL\phi)^2\D\ub'\lesssim |u|^{-1}\Omega_0^{-2}\int_0^{\ub}(rL\phi-\varphi)^2\D\ub'+\ub|u|^{-1}\Omega_0^{2}K^2,$$
and hence (using $|h_0|\le1$)
$$|h|\lesssim  |u|^{-1}\Omega_0^{-2}\int_0^{\ub}(rL\phi-\varphi)^2\D\ub'+\ub|u|^{-1}\Omega_0^{2}K^2+1,$$
plugging in we will have (using $\Omega_0^2\le |u|^p$ and $a\ge1$)
$$\left|\frac{\partial}{\partial u}|rL\phi(\ub, u)-\varphi(u)|^2\right|\lesssim c_K \ub|u|^{p-1}a \cdot|u|^{p-1}a |u|^p+  |u|^{p-2}a\int_{0}^{\ub}(rL\phi-\varphi)^2\D \ub'.$$
  Then (integrating $\ub\in[0,\ub_1]$)
$$\left|\frac{\partial}{\partial u}\int_0^{\ub_1}|rL\phi(\ub,u)-\varphi(u)|^2\D\ub\right|\lesssim c_K(\ub_1|u|^{p-1}a)^2 |u|^p+\ub_1|u|^{p-1}a\cdot|u|^{-1}\int_0^{\ub_1}|rL\phi(\ub,u)-\varphi(u)|^2\D\ub,$$
by \eqref{smallness}, if $\varepsilon$ is small enough, we have for some $c$ depending only on $K$, 
\begin{align*}&\int_{0}^{\ub_1}|rL\phi(\ub,u_1)-\varphi(u_1)|^2\D\ub\\
\ge &\frac{1}{2}\int_{0}^{\ub_1}|rL\phi(\ub,u_0)-\varphi(u_0)|^2\D\ub- c(\ub_1|u_1|^{p-1}a)^2\cdot|u_1|^{p+1}\\\ge &(\frac{1}{2}c_1-c\varepsilon^3)|u_1|^{\widetilde{p}+1}\end{align*}
where we have plugging in \eqref{deltau1}. Finally, by \eqref{psivarphibound} again, 
\begin{align*}\int_{0}^{\ub_1}|rL\phi(\ub,u_1)|^2\D\ub\ge& \frac{1}{2}(\frac{1}{2}c_1-c\varepsilon^2)|u_1|^{1+\widetilde{p}}-\ub_1\Omega_0^4(u_1)K^2\\
\ge&\frac{1}{2}(\frac{1}{2}c_1-c\varepsilon^2-K^2\varepsilon)|u_1|^{1+\widetilde{p}}.
\end{align*}
Let us choose $c_1$ such that $\frac{1}{2}(\frac{1}{2}c_1-c\varepsilon^2-K^2\varepsilon)\ge 2^4$. By \eqref{Dh} and \eqref{Omegaupper2}, evaluating at $(\ub_1,u_1)$, 
$$h-h_0\le -2^{-3} \Omega_0^{-2}|u_1|^{-1}\int_0^{\ub_1} (rL\phi)^2\D\ub\le -2.$$
Then $h\le-1<0$ and $S_{\ub_1,u_1}$ is trapped.

 \end{proof}

Now suppose that $rL\phi$ on $C_{u_0}$ assumes the form
\begin{equation}\label{Holderform} rL\phi(\ub,u_0)-\varphi(u_0)= t \ub^{\frac{p}{1-p}}, 0\le \ub\le\ub_1,\end{equation}
and we impose in addition that $\ub$ on $C_{u_0}$ to be the affine parameter. 
Then $\omega=0$ on $C_{u_0}$ and by \eqref{psivarphibound}, $a$ can be chosen simply to be $\max\{|t|, K\}$. The integral of  \eqref{lowerbound} becomes
\begin{align*}\int_0^{\ub_1}|rL\phi\big|_{C_{u_0}}-\varphi(u_0)|^2\D\ub=\frac{1-p}{1+p}t^2\ub_1^{\frac{1+p}{1-p}}=&\frac{1-p}{1+p}t^2(\varepsilon a^{-1})^{\frac{1+p}{1-p}}|u_1|^{\left(\frac{\widetilde{p}-p}{2}+1-p\right){\frac{1+p}{1-p}}}\\
=&\frac{1-p}{1+p}t^2(\varepsilon a^{-1})^{\frac{1+p}{1-p}}|u_1|^{1+\widetilde{p}+\frac{1}{2}(\widetilde{p}-p)\frac{3p-1}{1-p}}. \end{align*}
Then by Theorem \ref{fots},  $S_{\ub_1,u_1}$ is a closed trapped surface if
\begin{equation}\label{tlowerbound} t^2\ge \frac{1+p}{1-p}c_1(\varepsilon^{-1} a)^{\frac{1+p}{1-p}}|u_1|^{\frac{1}{2}(\widetilde{p}-p)\frac{1-3p}{1-p}}.\end{equation}
We can choose $a=\max\{|t|, K\}$,  and (note that $c_1$ depends on $K$)
\begin{equation}\label{deft}
|t|=\begin{cases}\left( \frac{1+p}{1-p}c_1(\varepsilon^{-1}K)^{\frac{1+p}{1-p}} |u_1|^{\frac{1}{2}(\widetilde{p}-p)\frac{1-3p}{1-p}}\right)^{\frac{1}{2}}, &  \frac{1+p}{1-p}c_1(\varepsilon^{-1}K)^{\frac{1+p}{1-p}} |u_1|^{\frac{1}{2}(\widetilde{p}-p)\frac{1-3p}{1-p}}\le K^{2},\\
\max\left\{K,\left( \frac{1+p}{1-p}c_1\varepsilon^{-\frac{1+p}{1-p}} |u_1|^{\frac{1}{2}(\widetilde{p}-p)\frac{1-3p}{1-p}}\right)^{\frac{1-p}{1-3p}}\right\}, &  \frac{1+p}{1-p}c_1(\varepsilon^{-1}K)^{\frac{1+p}{1-p}} |u_1|^{\frac{1}{2}(\widetilde{p}-p)\frac{1-3p}{1-p}}> K^{2}.
\end{cases}
\end{equation}

Let us summarize the above arguments as follows.
\begin{theorem}\label{fots2}
Consider the initial value problem on $C_{u_0}\cup \Cb_0$ in double null coordinates $(\ub,u)$. Assume $u_0\le-1$,  $\Omega\big|_{C_{u_0}}=\text{const}\le 1$ and there is some $u_1\in (-1,0)$ and $p>0$ such that for all $u\in (u_0,u_1)$,
$$\Omega_0(u)\le|u|^{p},$$
and for another $\widetilde{p}\ge p$,
$$\Omega_0(u_1)\le |u_1|^{\widetilde{p}}.$$
Suppose also that
$$|\psi|\le 1, |(\Omega_0^2h_0)^{-1}\varphi|\le K$$
for some $K\ge1$. We further assume
$$rL\phi(\ub,u_0)-\varphi(u_0)=t\ub^{\frac{p}{1-p}}, 0\le \ub\le \ub_1,$$
where $\ub_1$ is defined through \eqref{deltau1}, in which $\varepsilon$ is a constant depending on $K$, and $t$ verifies \eqref{tlowerbound} (and in particular can be to be  \eqref{deft}).  Then $S_{\ub_1,u_1}$ is a closed trapped surface.
\end{theorem}

In particular, if in addition $\widetilde{p}>p$, then $t$, considered as a function of $u_1$, tends to zero as $u_1\to0$, which refers to instability.

\begin{remark}To construct $C^1$ (or even smooth) perturbations (but close to the original data only in $C^{\frac{p}{1-p}}$), we can choose a larger $|t|$ and then make a cut off around $\ub=0$ to make $rL\phi$ smooth in $[0,\ub_1]$ so that \eqref{lowerbound} still holds.  See the last section of \cite{L-Z3}.
 \end{remark}

\subsection{Proof of Theorem \ref{thm:exterior}: Exterior instability}

 Let us first fix a $k$ with $0<k^2<\frac{1}{3}$.  We will apply Theorem \ref{fots2} with $\Cb_0=\{s=s_*\}$, the past null cone of the singularity $\mathcal{O}$, and $C_{u_0}$ be a outgoing null cone starting from somewhere of $\Cb_0$, such that $u_0\le -1$. 

On $s=s_*$ we have $\theta=\frac{1}{k}$ and $\zeta=-k$ by \eqref{thetas*} and \eqref{zetas*}. Converted to double null quantities, we have
$$\varphi(u)=\Omega_0^2(u)h_0(u)\frac{1}{k}, \psi(u)=k,$$
for any $u\in[u_0,0)$, which verifies the condition \eqref{psivarphibound} with $K=\frac{1}{k}$. The equation in \eqref{Omega_0} then implies that, recalling $\Omega_0(u_0)\le 1$, for all $u\in (u_0,0)$,
\begin{equation}\label{Omegaupper-ext}\Omega_0^2(u)\le |u|^{k^2}.\end{equation}
We choose $\widetilde{p}=k^2$, and $p<\widetilde{p}$ is arbitrarily chosen. Let us take
$$rL\phi(\ub,u_0)-\varphi(u_0)=t\ub^{\frac{p}{1-p}}, 0\le\ub\le\ub_1$$
 where $\ub_1$ is defined through \eqref{deltau1} and $t$ is defined by \eqref{deft}. Then $S_{\ub_1,u_1}$ is a closed trapped surface by Theorem \ref{fots2}. Because \eqref{Omegaupper-ext} holds for all $u\in (u_0,0)$, we can choose an arbitrary $u_1\in (u_0,0)$ and view $t$ as a function of $u_1$. Then $rL\phi$ can be viewed as a family of initial data on $C_{u_0}$ with parameter $u_1$. When $u_1\to0$,
$$\frac{1+p}{1-p}(\varepsilon^{-1}K)^{\frac{1+p}{1-p}} |u_1|^{\frac{1}{2}(\widetilde{p}-p)\frac{1-3p}{1-p}}\to 0,$$
and hence $t\to0$ (taking the first choice in \eqref{deft}). The family of  initial data $rL\phi$ then approaches the original initial data in $C^{\frac{p}{1-p}}_{\ub}$. Note that $p$ can be arbitrarily chosen such that $p<\widetilde{p}=k^2$, this gives the   exterior instability, proving Theorem \ref{thm:exterior}.

\begin{remark}
Note that for any $p<\widetilde{p}=k^2$ and any $t\ne0$, \eqref{tlowerbound} holds for $u_1$ sufficiently close to $0$.  Since \eqref{Omegaupper-ext} holds for all $u\in (u_0,0)$, for any fixed $t\ne0$, $S_{\ub_1,u_1}$ is trapped for all $u_1$ sufficiently close to $0$ (and correspondingly, for all $\ub_1$ sufficiently small), more than just one single trapped surface. This is the same as the exterior perturbations studied in \cite{Chr99, Liu-Li}. The analysis in \cite{Chr91} implies that the whole future boundary of the resulting maximal future development is the future boundary of the  trapped region and is strictly spacelike, verifying the \emph{strong cosmic censorship conjecture} (see \cite{Chr99cqg} for more discussions).  As can be seen in the following proof of Theorem \ref{thm:interior}, the above argument does not work for interior perturbations. It is then not known whether there exist interior perturbations in $C^{\frac{p}{1-p}}$ leading to a maximal future development that has a strictly spacelike future boundary. The future boundary may have a small piece that is the regular future null cone of the singularity.  But as a trapped surface is still showed to form, the future boundary enters the trapped region and becomes spacelike eventually. 
\end{remark}

\begin{remark}
It is also interesting to see what happen at the threshold $p=k^2$, in which case $\widetilde{p}=p$ and Theorem \ref{fots2} can also be applied.  In view of \eqref{tlowerbound} and \eqref{deft}, a trapped surface can form if
$$|t|\ge\max\left\{\frac{1}{k},\left( \frac{1+k^2}{1-k^2}c_1\varepsilon^{-\frac{1+p}{1-p}} \right)^{\frac{1-p}{1-3p}}\right\}.$$
The exterior nonlinear stability was provided by Singh in \cite{JS2} at the threshold. So a trapped surface can form only when the perturbations are large at threshold regularity (which can however be viewed as a small perturbation at a regularity below the threshold).  This can be compared to the trapped surface formation result by An-Luk \cite{An-Luk} under a large scale-critical (viewed as the threshold in Minkowski background) incoming energy. 
\end{remark}

\subsection{Proof of Theorem \ref{thm:interior}: Interior instability}

As mentioned in Introduction, if we turn to interior instability, we cannot directly apply  Theorem \ref{fots2} using a single incoming null cone in the interior region, because the the incoming null cones in the  interior region have regular center. The conditions \eqref{Omegaupper} and \eqref{Omegaupper2} cannot hold for arbitrarily small $|u_1|$. Then in view of the condition \eqref{tlowerbound}, $t$ cannot be chosen arbitrarily small and hence no instability can be found in this way. To overcome this difficulty, our new approach  is to consider a family of interior null cones approaching the last one (the past null   cone of $\mathcal{O}$). Thanks to the $k$-self-similarity, the geometry of these null cones are particularly clear.

 Let $(u,r)$ be the self-similar Bondi coordinates in the $k$-self-similar solution. We fix an outgoing null cone, say $u=-1$, denoted by $C_{-1}^+$, which is viewed as the initial null cone.  We denote $C^-_s$, to be the incoming null cone in the interior region such that the self-similar parameter takes value $s$ at its intersection with $C_{-1}^+$ (see Figure \ref{fig:nakedinterior2}).  As proven in \cite{Chr94}, for $s<s_*$ these null cones reach the center. By scaling $u$ by a constant, we assume that the intersection of $C_{-1}^+$ to the last incoming null cone  $C_{-1}^+\cap C_{s_*}^-$ has area radius larger than $1$, and such that for $s$ sufficiently close to $s_*$, the area radius of $C_{-1}^+\cap C_{s}^-$ is still larger than $1$.  For different $s$, we will then apply Theorem \ref{fots2} to $\Cb_{0,s}=C^-_s$, $C_{u_{0,s},s}=C_{-1}^+$, with the associated double null coordinates $(\ub_s,u_s)$. Note that all $C_{u_{0,s},s}$ are the same null cone, on which $\ub_s$ is affine parameter with $\Omega_{0,s}(u_{0,s})\le 1$. We make a remark on notations that $u$ is the self-similar Bondi coordinate and $u_s$ is the optical function used in the Theorems before. Another observation is that  $u_s$ is equal to $-r$ on $\Cb_{0,s}$, which is   a coordinate independent quantity.  Moreover, we have $u_{0,s}\le -1$.

 \begin{figure}
\includegraphics[width=1.8 in]{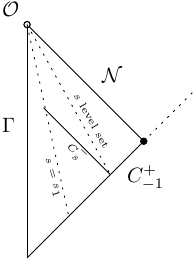}
\caption{}
\label{fig:nakedinterior2}
\end{figure}

Recall that $\zeta$ tends to $-k$ as $s\to s_*$. Therefore for any $\gamma>0$, we can always find some $s_1<s_*$ such that for any $s\in (s_1,s_*)$,
$$\zeta^2\ge( 1-\gamma)k^2.$$
Converted to the notations (with an additional subscript $s$) in our theorems,  it means that 
\begin{equation}\label{psiclosetok}\psi_s(u_s)^2\ge(1-\gamma) k^2,\end{equation}
for any $u_s\in[u_{0,s},u_{1,s}]$, where $-u_{1,s}$ equal to the area radius $r_s$ of $C^-_s\cap\{s=s_1\}$. Then the integral in \eqref{Omega_0} satisfies, for $u_s\in[u_{0,s},u_{1,s}]$,
\begin{equation*}\int_{u_{0,s}}^{u_s}\frac{\psi_s^2(u')}{|u'|}\D u'\ge -k^2(1-\gamma)\log\frac{|u_s|}{|u_{0,s}|},\end{equation*}
which then implies that for $u_s\in [u_{0,s},u_{1,s}]$,  recalling $\Omega_{0,s}(u_{0,s})\le 1$ and $u_{0,s}\le -1$, 
\begin{equation}\label{Omegaupperselfsimilar}\Omega_{0,s}^2(u_s)\le |u_s|^{(1-\gamma)k^2}.\end{equation}
By choosing possibly closer $s_1$, one may also have, for $s\in (s_1,s_*]$,
$$|\zeta|\le 2k, |\theta|\le \frac{2}{k},$$
which is, for $u_s\in [u_{0,s},u_{1,s}]$, 
\begin{equation}\label{psivarphiboundselfsimilar}|\psi_s|\le 2k, |(\Omega_{0,s}^2h_{0,s})^{-1}\varphi_k|\le\frac{2}{k}.\end{equation}
We choose $\widetilde{p}=(1-\gamma)k^2$ and $p<\widetilde{p}$ is arbitrarily chosen.  $\widetilde{p}$ and $p$ can be chosen independent of $s$. Taking $K=\frac{2}{k}$ we can apply Theorem \ref{fots2}. Let us take
$$(rL\phi)_s(\ub_s,u_{0,s})-\varphi_s(u_{0,s})=t_s\ub_s^{\frac{p}{1-p}}, 0\le\ub_s\le \ub_{1,s},$$
where $\ub_{1,s}$ is defined through \eqref{deltau1} and $t_s$ is defined by \eqref{deft}. Then a trapped surface froms at $S_{\ub_{1,s},u_{1,s}}$. In order to show that $t_s\to0$ as $s\to s_*$, we need to show that $|u_{1,s}|\to0$.

Consider the equation of the incoming null cone (ray) $C^-_s$, which is, according to the metric form \eqref{Bondimetric},
$$\frac{\D r}{\D u}=-\frac{1}{2}\mathrm{e}^{\nu-\lambda}.$$
Recalling that $s=\log(-\frac{r}{u})$, we will have, along $C_s^-$,
$$\frac{\D s}{\D u}=-\frac{1}{2r}\mathrm{e}^{\nu-\lambda}-\frac{1}{u},$$
therefore
\begin{equation*}\frac{\D r}{\D s}=\frac{r(\beta-1)}{\beta}=r(1-\alpha).\end{equation*}
Integrate the above equation from $s$ to $s_1$, we have
\begin{equation*} 
\log (-u_{1,s}) -\log (r\big|_{C^+_{-1}\cap C^-_s})=\int_s^{s_1}(1-\alpha)\D s'\sim \frac{1}{1-k^2}\log(s_*-s), s\to s_*.
\end{equation*}
Because $\log (r\big|_{C^+_{-1}\cap C^-_s})=s$, we have
\begin{equation}\label{estimateu1s}
\log (-u_{1,s}) -s_1=\int_s^{s_1}(-\alpha)\D s'\sim \frac{1}{1-k^2}\log(s_*-s), s\to s_*,
\end{equation}
in view of \eqref{alphas*}. This proves $|u_{1,s}|\to0$ at a precise rate.  

Finally, we transform the above perturbations into one fixed coordinate system.  We fix an affine parameter $v$ on $C_{-1}^+$ in the way that $v$ is a function of $r$ on $C_{-1}^+$ where $r=\mathrm{e}^{s}$ (recalling that $s=\log(-\frac{r}{u})$) and denote $r_*=\mathrm{e}^{s_*}$. We define $v$ through the initial condition 
$$v(r_*)=v_*, v'(r_*)=2.$$
For each $s$ sufficiently close to $s_*$, we choose
$$\ub_s=v-v_s$$
where $v_s=v(\mathrm{e}^s)$  is the value of $v$ at $C_s^-\cap C_{-1}^+=\Cb_{0,s}\cap C_{u_{0,s},s}$, on which $\ub_s=0$. 
We compute $\Omega_{0,s}(u_{0,s})$ to show that it is less than $1$ when $s$ sufficiently close to $s_*$. Recall that
\begin{align*}-2\Omega_{0,s}^2(u_{0,s})=&g\left(\frac{\partial}{\partial u_s}, \frac{\partial}{\partial \ub_s}\right)\big|_{\ub_s=0, u_s=u_{0,s}}\\
=&g\left(2\mathrm{e}^{\lambda-\nu}\frac{\partial}{\partial u}-\frac{\partial}{\partial r}, \frac{\partial}{\partial v}\right)\big|_{C_s^-\cap C_{-1}^+}=-2\mathrm{e}^{2\lambda(s)}\frac{\partial r}{\partial v}\big|_{C_s^-\cap C_{-1}^+}\end{align*}
by the expression of the metric in self-similar Bondi coordinates \eqref{Bondimetric} and $r$ can also be viewed as a function of $v$ on $C_{-1}^+$. When $s_1$ sufficiently close to $s_*$, we have $\frac{\partial r}{\partial v}\big|_{C_s^-\cap C_{-1}^+}\le\frac{1}{1+2k^2}$(which is larger than $\frac{1}{2}=\frac{\partial r}{\partial v}\big|_{C_{s_*}^-\cap C_{-1}^+}$) for all $s\in(s_1,s_*)$, then
$$\Omega^2_{0,s}(u_{0,s})\le \frac{1}{1+2k^2} \mathrm{e}^{2\lambda(s)}.$$
From \eqref{solvelambda} we have $\mathrm{e}^{2\lambda(s)}\to 1+k^2$ as $s\to s_*$, so when $s_1$ is sufficiently close to $s_*$, we will have $\Omega^2_{0,s}(u_{0,s})\le1$ for all $s\in (s_1,s_*)$. The perturbation, as a family (in $s\in (s_1,s_*)$) of function defined on the whole $C^+_{-1}$ starting from the center, is then in the form
\begin{equation}\label{formofperturbation}r\partial_v\phi(v;s)=\begin{cases}r\partial_v\phi_k(v), & v(0)\le v< v_s,\\ r\partial_v\phi_k(v_s)+t_s(v-v_s)^{\frac{p}{1-p}}, &v_s\le v< v_s+\ub_{1,s},\\
\text{connected continuously by linear function}, &v_s+\ub_{1,s}\le v< v_s+2\ub_{1,s},\\
r\partial_v\phi_k(v), &v\ge v_s+2\ub_{1,s}. \end{cases}\end{equation}
 where $\phi$ is the perturbed wave function and $\phi_k$ is the original wave function (with an asymptotically flat cut off for large $v$) in the $k$-self-similar naked singularity solution. Because $t_s\to0$, this proves $r\partial_v\phi(\cdot;s)$ converges to $r\partial_v\phi_k$ in $C^{\frac{p}{1-p}}_v$ (and also in AC). Since $\gamma$ in \eqref{psiclosetok} can be chosen arbitrarily small ($s_1$ changes as well), and $p$ can be arbitrarily close to $\widetilde{p}=(1-\gamma)k^2$. This proves  interior instability below the threshold.

At last we want to show that $v=v_s+2\ub_{1,s}< v_*=v(r_*)$ for $s$ sufficiently close to $s_*$, which means that the perturbations completely lie in the original interior region. To see this, we use \eqref{deltau1} and \eqref{estimateu1s}. From \eqref{estimateu1s}, for $s$ sufficiently close to $s_*$, we have, 
$$|u_{1,s}|\mathrm{e}^{-s_1}\le (s_*-s)^{\frac{1}{1-(1-\gamma)k^2}}.$$
Then for $s$ sufficiently close to $s_*$,
$$2\ub_{1,s}\le 2\varepsilon a_s^{-1} |u_{1,s}|^{\frac{(1-\gamma)k^2-p}{3}+1-p}\le 2\varepsilon\mathrm{e}^{s_1\left({\frac{(1-\gamma)k^2-p}{3}+1-p}\right)}(s_*-s)^{\frac{1}{1-(1-\gamma)k^2}\left({\frac{(1-\gamma)k^2-p}{3}+1-p}\right)}.$$
The power of $(s_*-s)$ is larger than $1$ if $p<(1-\gamma)k^2$.  Because $\frac{\partial r}{\partial v}$ is uniformly bounded  near $s_*$, hence 
$$v_*-v \approx r_*-r\approx s_*-s,$$
which shows that $v_s+2\ub_{1,s}<v_*$ for $s$ sufficiently close to $s_*$.  Finally, the above relation allows us to write the perturbation \eqref{formofperturbation} in the coordinate $r$ and also in terms of $\partial_r(r\phi)$. This proves Theorem \ref{thm:interior}. We can see moreover that the trapped surface $S_{\ub_{1,s},u_{1,s}}$ is located in the original interior region.

\begin{remark}The numbers $\gamma,\widetilde{p}$ and $p$ can be chosen depending on $s$, so that we can construct a single family of perturbations in any regularity below the threshold.
\end{remark}

\section{Proof of Theorem \ref{thm:BVinterior}: Interior instability of general naked singularities}

At last we deal with general naked singularity solutions. Let us apply  Theorem \ref{estimate} for $p=0$ (which is exactly that had been established in \cite{Liu-Li}).  As long as \eqref{smallness} holds, estimate \eqref{proof-estimate-Lphiimp} holds without $\varepsilon^{-\delta}$, then we always have
\begin{equation}\label{lowerboundLphi}|rL\phi(\ub,u)|\ge |rL\phi(\ub,u_0)-\varphi(u_0)+\varphi(u)|- c\ub|u|^{-1}w_1^2a^2\end{equation}
for some $c>0$. On the other hand, using \eqref{Omega_0} we observe that 
\begin{equation}\label{estimate-w}w^2(u)\le|\log\Omega_0(u)|, u\in [u_0,u_1].\end{equation}
Then we take $\Cb_{0,n}=\mathcal{N}_n$ and $C_{u_{0,n},n}=C_{-1}^+$ with  associated double null coordinates $(\ub_{n}, u_n)$ to apply Theorem \ref{estimate}, then  \eqref{lowerboundLphi} holds for every $n$. We consider two separated cases that $\varphi_n(u)$ is uniformly bounded in both $n$ and $u_n\in(u_{0,n},0)$ or not.

Case 1: Suppose that (possibly passing to a subsequence) there exists $u_{1,n}$ such that $|\varphi_n(u_{1,n})|\to\infty$. We in addition choose $u_{1,n}$ such that 
$$|\varphi_n(u_{1,n})|=\sup_{u_{0,n}\le u_n\le u_{1,n}}|\varphi_n(u_n)|\to\infty.$$
  It is not necessary that $u_{1,n}\to0$. Let us apply the estimate \eqref{lowerboundLphi} to each $n$. Set $\ub_{1,n}$ as
\begin{equation}\label{case1ub1n}\ub_{1,n}|u_{1,n}|^{-1}|\log\Omega_{0,n}(u_{1,n})||\varphi_n(u_{1,n})|=\varepsilon\end{equation} 
where $\varepsilon$ is the constant determined in Theorem \ref{estimate}.  Choose $n$ sufficiently large such that for any $\ub_n\in[0,\ub_{1,n}]$,
$$|(rL\phi)_n(\ub_n,u_{0,n})-\varphi(u_{0,n})|\le \frac{1}{2}|\varphi_n(u_{1,n})|$$
 and hence $a_n= |\varphi_n(u_{1,n})|$ for sufficiently large $n$. Then \eqref{smallness} holds and  we will have, for $\ub_n\in[0,\ub_{1,n}]$, 
\begin{align*}|(rL\phi)_n(\ub_n,u_{1,n})|\ge&|\varphi_n(u_{1,n})|- |(rL\phi)_n(\ub_n,u_{0,n})-\varphi(u_{0,n})|-c\ub_n|u_{1,n}|^{-1}w_{1,n}^2 |\varphi_n(u_{1,n})|^2\\
\ge&(\frac{1}{2}-c\varepsilon)|\varphi_n(u_{1,n})|\ge\frac{1}{4}|\varphi_n(u_{1,n})|\end{align*}
if we require in addition $4c\varepsilon \le 1$.  We will then have, using \eqref{case1ub1n},
$$\int_0^{\ub_{1,n}}|(rL\phi)_n(\ub_n,u_{1,n})|^2\D\ub_n\ge\frac{1}{16}\ub_{1,n}|\varphi_n(u_{1,n})|^2=\frac{\varepsilon}{16}\frac{|\varphi_n(u_{1,n})|}{|\log\Omega_{0,n}(u_{1,n})|}|u_{1,n}|\ge 2^4\Omega^2_{0,n}(u_{1,n})|u_{1,n}|$$
for sufficiently large $n$. Then by \eqref{Dh}, 
$$h(\ub_{1,n}, u_{1,n})-h_0(u_{1,n})\le -2^{-3} \Omega_{0,n}^{-2}(u_{1,n})|u_{1,n}|^{-1}\int_0^{\ub_{1,n}} (rL\phi)_n^2\D\ub_n\le -2.$$
Then $h(\ub_{1,n},u_{1,n})<0$ and $S_{\ub_{1,n},u_{1,n}}$ is trapped. Since the interior region does not contain trapped surfaces, logically this includes two possibilities, one of which is that $S_{\ub_{1,n},u_{1,n}}$ is located in the exterior and the singularity is then not naked, and the other of which is that $|\varphi_n(u)|$ is uniformly bounded for all $n$.

Case 2: There exists a $K\ge1$ such that $|\varphi_n(u_n)|\le K$ for all $u_n\in[u_{0,n}, 0)$. 
Now we will use the fact that 
$$\Omega_{0,n}^2(0)\to 0, n\to\infty$$
which is a consequence of \eqref{blueshfit4} and \eqref{Omega_0}.  As $\Omega_{0,n}^2(u)$ is monotonically decreasing in $u\to0$, we can find  an increasing family $u_{1,n}\to 0$ so that $\Omega_{0,n}^2(u_{1,n})\to0$ as $n\to\infty$.  In this case we set
\begin{equation}\label{case2ub1n}\ub_{1,n}=\Omega_{0,n}^{\frac{2}{3}}(u_{1,n})|u_{1,n}|.\end{equation}
We fix a constant, say  $\gamma=\frac{1}{10}$. We also assume (maybe for larger $K$)
$$|(rL\phi)_n(\ub_n, u_{0,n})-\varphi_n(u_{0,n})|\le K, 0\le \ub_n\le \ub_1$$
for all $n$ since it is about the data of the given solution.  Then $a_n\le K$ for all $n$.  By \eqref{estimate-w}, \eqref{case2ub1n} implies  that \eqref{smallness} holds for  sufficiently large $n$. Therefore we can apply \eqref{lowerboundLphi} for sufficiently large $n$.

 Suppose that
\begin{equation}\label{lowerboundgeneral}\int_0^{\ub_{1,n}}|(rL\phi)_n(\ub_n, u_{0,n})-\varphi_n(u_{0,n})+\varphi_n(u_{1,n})|^2\D\ub_n> \Omega_{0,n}^{2-2\gamma}(u_{1,n})|u_{1,n}|\end{equation}
holds for a subsequence (still denoted by subscript $n$). Applying  \eqref{lowerboundLphi} we have
\begin{align*}|(rL\phi)_n(\ub_n,u_{1,n})|\ge& |(rL\phi)_n(\ub_n,u_{0,n})-\varphi_n(u_{0,n})+\varphi_n(u_{1,n})|- c\ub_{n}|u_{1,n}|^{-1}w_{1,n}^2K^2\\
\ge& |(rL\phi)_n(\ub_n,u_{0,n})-\varphi_n(u_{0,n})+\varphi_n(u_{1,n})|- c\Omega_{0,n}^{\frac{2}{3}}(u_{1,n})w_{1,n}^2K^2.\end{align*}
For sufficiently large $n$, $c\Omega_{0,n}^{\frac{2}{3}}(u_{1,n})w_{1,n}^2 K^2\le \frac{1}{2}\Omega_{0,n}^{\frac{2}{3}-\gamma}(u_{1,n})$ where we have used \eqref{estimate-w} to control $w_{1,n}$, then
$$\int_{0}^{\ub_{1,n}}|(rL\phi)_n(\ub_n,u_{1,n})|^2\D\ub_n>\frac{1}{4}\Omega_{0,n}^{2-2\gamma}(u_{1,n})|u_{1,n}|\ge 2^4\Omega^2_{0,n}(u_{1,n})|u_{1,n}|$$
for sufficiently large $n$. As in Case 1, $S_{\ub_{1,n},u_{1,n}}$ is trapped. 

Suppose at last that
$$\int_0^{\ub_{1,n}}|(rL\phi)_n(\ub_n, u_{0,n})-\varphi_n(u_{0,n})+\varphi_n(u_{1,n})|^2\D\ub_n\le \Omega_{0,n}^{2-2\gamma}(u_{1,n})|u_{1,n}|$$
for any $n$, we will show that can perturb the initial data so that ``$\le$'' becomes ``$>$'', that is \eqref{lowerboundgeneral} holds for sufficiently large $n$, and hence a trapped surface will form. In fact, we set
\begin{equation}\label{Case2perturb}(rL\phi)_n^{\text{perturb}}(\ub_n,u_{0,n})= (rL\phi)_n(\ub_n,u_{0,n})+\frac{\Omega_{0,n}^{\frac{2}{3}-2\gamma}(u_{1,n})}{\ub_{1,n}}\ub_n, 0\le\ub_n\le \ub_{1,n}.\end{equation}
Note that this perturbation will not change $\varphi_n$, so
\begin{align*}&\int_0^{\ub_{1,n}}|(rL\phi)_n^{\text{perturb}}(\ub_n, u_{0,n})-\varphi_n(u_{0,n})+\varphi_n(u_{1,n})|^2\D\ub_n\\
\ge&\int_0^{\ub_{1,n}}\left|\frac{\Omega_{0,n}^{\frac{2}{3}-2\gamma}(u_{1,n})}{\ub_{1,n}}\ub_n\right|^2\D\ub_n-\frac{1}{2}\int_0^{\ub_{1,n}}|(rL\phi)_n(\ub_n, u_{0,n})-\varphi_n(u_{0,n})+\varphi_n(u_{1,n})|^2\D\ub_n\\
\ge&\frac{1}{3}\Omega_{0,n}^{2-4\gamma}(u_{1,n})|u_{1,n}| -\frac{1}{2} \Omega_{0,n}^{2-2\gamma}(u_{1,n})|u_{1,n}|\\
>& \Omega_{0,n}^{2-2\gamma}(u_{1,n})|u_{1,n}|\end{align*}
for sufficiently large $n$. As in  the case assuming \eqref{lowerboundgeneral},  $S_{\ub_{n,1},u_{1,n}}$ is trapped. 

We may extend the perturbations \eqref{Case2perturb} to $\ub_n<0$ just as the original data, and to $\ub_n>\ub_{1,n}$ linearly back to the original data after a fixed distance so that it is absolutely continuous. By translating back to the the fixed coordinate $r$ as in $k$-self-similar cases, the perturbed data converges to the original data in BV topology. This proves Theorem \ref{thm:BVinterior}. If we consider $C^1$ data and solutions, we can obtain $C^1$ perturbations (small in BV topology) by cutting off argument.

\end{document}